\documentclass[12pt,a4paper,oneside,reqno]{amsart}

\usepackage{amssymb,amsmath}
\usepackage{graphicx,psfrag}
\usepackage{stmaryrd}

\usepackage{color}

\definecolor{thmcolor}{rgb}{0,0,.4} 
\definecolor{remarkcolor}{rgb}{0,.2,0} 
\definecolor{proofcolor}{rgb}{.4,0,0} 
\definecolor{quecolor}{rgb}{.2,.2,0} 
\definecolor{axcolor}{rgb}{.3,0,.3}
\definecolor{thmbgcolor}{rgb}{0.9,0.9,1} 
\definecolor{rmbgcolor}{rgb}{0.9,1,0.9} 
\definecolor{proofbgcolor}{rgb}{1,0.9,0.9}

\theoremstyle{definition} \newtheorem{thm}{\colorbox{thmbgcolor}{\textcolor{thmcolor}{Theorem}}}[section] 
\theoremstyle{definition} \newtheorem{cor}[thm]{\colorbox{thmbgcolor}{\textcolor{thmcolor}{Corollary}}} 
\theoremstyle{definition} 
\theoremstyle{definition} \newtheorem{prop}[thm]{\colorbox{thmbgcolor}{\textcolor{thmcolor}{Proposition}}}
 
\theoremstyle{remark}  
\theoremstyle{remark} 
\theoremstyle{definition}  
\theoremstyle{definition} \newtheorem{rem}[thm]{\colorbox{rmbgcolor}{\textcolor{remarkcolor}{Remark}}}
\theoremstyle{definition}

\newcommand{\taxis}{\mathsf{t\text{-}axis}}

\newcommand{\timed}{\mathsf{time}}
\newcommand{\sqspace}{\mathsf{space}^2}

\newcommand{\ls}{\mathsf{c}}

\newcommand{\vx}{\mathbf{\bar x}}
\newcommand{\vy}{\mathbf{\bar y}}

\newcommand{\vu}{\mathbf{\bar u}}
\newcommand{\de}{:=}

\newcommand{\defiff}{\ \stackrel{\;def}{\Longleftrightarrow}\ }

\newcommand{\IOb}{\ensuremath{\mathsf{IOb}}} 
\newcommand{\IB}{\ensuremath{\mathsf{IB}}} 
\newcommand{\B}{\ensuremath{\mathit{B}}} 
\newcommand{\Ph}{\ensuremath{\mathsf{Ph}}} 
\newcommand{\LS}{\ensuremath{\mathsf{LS}}} 
\newcommand{\Q}{\ensuremath{\mathit{Q}}} 
\newcommand{\W}{\ensuremath{\mathsf{W}}} 
\newcommand{\ev}{\ensuremath{\mathsf{ev}}} 
\newcommand{\Id}{\ensuremath{\mathsf{Id}}} 
\newcommand{\ax}[1]{\textcolor{axcolor}{\ensuremath{\mathsf{#1}}}} 
\newcommand{\wl}{\ensuremath{\mathsf{wl}}}

\newcommand{\w}{\ensuremath{\mathsf{w}}}

\begin{document}

\title[The existence of FTL particles is consistent with SR]{The
  existence of superluminal particles is consistent with the kinematics of
  Einstein's special theory of relativity}
\author{Gergely Sz\'ekely}
\address{Alfr\'ed R\'enyi Institute of Mathematics\\
of the Hungarian Academy of Sciences\\
Budapest P.O.Box 127, H-1364, Hungary}
\email{szekely.gergely{@}renyi.mta.hu.}


\begin{abstract} 
Within an axiomatic framework of kinematics, we prove that the
existence of faster than light particles is logically independent of
Einstein's special theory of relativity. Consequently, it is
consistent with the kinematics of special relativity that there might
be faster than light particles.
\end{abstract}

\maketitle

\section{Introduction}

From time to time short-lived experimental results appear that suggest
the existence of FTL objects. Recently, the OPERA experiment, see
\cite{opera}, raised the interest in the possibility of FTL
particles. The trouble is that if there are FTL particles, then
several branches of the tree that grew out of relativity theory die
out, e.g., the ones that directly assume the nonexistence of
superluminal objects. Weinberg--Salam theory is a concrete example for a
  theory which have to be modified if someone discovers FTL particles,
  see M\'esz\'aros \cite{Meszaros1}, \cite{Meszaros2}.

The OPERA result has turned out to be erroneous, but the possibility
will always be there that one day an experiment will prove the
existence of FTL particles. Therefore, the question 
\begin{center}
\it Which parts of tree sprung out of relativity theory would
survive the discovery of FTL particles?
\end{center}
remained interesting and relevant for further investigation.  In
this paper, we axiomatically show that the roots of the metaphorical
tree surely endure any experiment proving the existence of FTL objects
since their existence is completely consistent with the kinematics of
special relativity.

The investigation of superluminal motion in relativity theory goes
back (at least) to Tolman, see \cite[p.54-55]{tolman}.\footnote{A
    detailed history of the tachyon concept tracing back even to
    pre-relativistic times can be found in \cite{Froman}.} After
showing that faster than light (FTL) particles travel back in time
according to some observers,\footnote{This observation is the basis of
  several causal paradoxes (i.e., seemingly contradictory statements)
  concerning FTL particles.} Tolman writes: ``Such a condition of
affairs might not be a logical impossibility; nevertheless its
extraordinary nature might incline us to believe that no causal
impulse can travel with a velocity greater than that of light.''  It
is interesting to note that Tolman has not claimed that relativity
theory implies the impossibility of the existence of superluminal
particles; he just claims that, if they exists, they have some
``extraordinary'' properties.

In 1962 Bilaniuk, Deshpande and Sudarshan, in their pioneering
article, introduce a reinterpretation principle suggesting that
superluminal particles are consistent with relativity theory
\cite{BilaniukDeshpandeSudarshan}. Since then a great many works
dealing with superluminal motion have appeared in the literature.  The
extensive survey of Recami reviews the papers dealing with
superluminal motion before 1986 \cite{Recami86}.  For more recent
papers concerning FTL motion, see, e.g., Arntzenius \cite{Arntzenius},
Chashchina-Silagadze~\cite{ChaSil}, Geroch \cite{Geroch},
Jentschura~\cite{Jentschura}, Jentschura--Wundt~\cite{JW12},
  Nikoli{\'c}~\cite{Nikolic}, Matolcsi--Rodrigues
\cite{matolcsi-ftl}, Mittelstaedt \cite{mittelstaedt-ftl}, Recami
\cite{recami-ftl}, \cite{recami-ftl2}, \cite{Recami09},
Selleri~\cite{selleri-ftl}, Recami--Fontana--Garavaglia
\cite{RFG-ftl}, Weinstein~\cite{weinstein-ftl},
Zamboni-Rached--Recami--Besieris~\cite{Zamboni-Rached} and
references therein.

These papers contain various non-axiomatic theories of FTL
particles compatible with relativity.  However, the only
framework where the question of consistence can properly be
investigated is the axiomatic framework of mathematical logic.
Therefore, in this paper, we take one step further and
  investigate the consistency question of FTL particles within
  mathematical logic.

Based on Einstein's original postulates, we formalize the kinematics
of special relativity within an axiomatic framework; and we
prove that axiom system \ax{SR} (see, p.\pageref{SR}) capturing the
kinematics of Einstein's special relativity does not contradict the
existence of FTL particles.

The fact that the axioms of relativistic kinematics do not contradict
the existence of FTL particles is interesting by itself. However, it
leaves the question open whether the axioms of relativistic dynamics
contradicts superluminal motion of particles or not.  In a forthcoming
paper, we will show, with the axiomatic rigor of this paper,
that even relativistic particle dynamics is consistent with the
existence of FTL particles, see \cite{MSzFTL}, as it is
  already suggested by the literature.

As general relativity is more general than special relativity, every
model (solution) of the axioms of special relativity is also a
model of the axioms of general relativity.\footnote{For an axiom
  system of general relativity explicitly reflecting this fact, see
  \cite{synthese}.}  Consequently, if special relativity has a model
allowing FTL particles, then general relativity also has such a
model. Therefore, our result implies that, the existence of
superluminal particles is consistent with general relativity.
  
We show that the statement ``there can be faster than light
particles'' is logically independent of the kinematics of special
relativity. This means that we can add either this assumption or its
opposite to the axioms of special relativity without getting a theory
containing contradictions.  

This logical independence is completely analogous to that
Euclid's postulate of parallels is independent of the rest of his
axioms (in this case two different consistent theories extending the
theory of absolute geometry are Euclidean geometry and hyperbolic
geometry).

First in Section~\ref{sec-inf}, we explain our result and axiomatic
framework without going into the details of formalization. Then in
Section~\ref{sec-lang}, we recall a first-order logic framework of
kinematics from the literature with a minor modification fitting it to
formalize Einstein's original postulates. Then in
Section~\ref{sec-ax}, we formulate Einstein's original two postulates
and supplement them with some natural ones (e.g., with the statement
that different inertial observers see the same events) which were
implicitly assumed by Einstein, too. It is also a benefit of using
first-order logic that we have to reveal all our tacit assumptions.

Within this axiomatic framework, we formulate and prove our main
result, namely that the existence of FTL particles is logically
independent of the kinematics of special relativity, i.e., we prove
that neither the existence nor the nonexistence of FTL objects follows
from the theory, see Theorem~\ref{thm-indep}. Consequently, it is
consistent with the kinematics of special relativity that there are
FTL particles, which can of course carry ``information,'' see, e.g.,
Section \ref{sec-ftl}. In Section \ref{sec-no}, we will discuss why
the possibility of sending particles back to the past does not
necessarily lead to a logical contradiction even if we leave our safe
framework of kinematics.

\section{Informal statement of the main result}
\label{sec-inf}
Before we get wrapped up in details of formalization, let us first
state our main result informally. We assume the formalized versions of
Einstein's two original postulates of \cite{einstein}:
\begin{itemize}
\item principle of relativity (see \ax{SPR} on p.\pageref{spr}), and
\item the light axiom (see \ax{AxLight} on p.\pageref{axlight}).
\end{itemize}
To see the boundaries of special relativity 
clearly, we list all the other assumptions. Thees assumptions are
tacitly there in all the approaches to special relativity:
\begin{itemize}
\item Physical quantities satisfy some nice properties of real numbers (see \ax{AxOField} on p.\pageref{axofield}).
\item Inertial observers coordinatize the same outside reality (see \ax{AxEv} on p.\pageref{axev}).
\item Inertial observers can move with any speed slower than that of light (see \ax{AxThExp} on p.\pageref{axthexp}).
\item Inertial observers are stationary according to themselves (see \ax{AxSelf} on p.\pageref{axself}).
\item Inertial observers (can) use the same units of measurements (see \ax{AxSymD} on p.\pageref{axsymd}).
\end{itemize}
Because it is not at all clear which assumptions are unquestionable, we do not distinguish between axioms and postulates, rather we follow the modern axiomatic approach, i.e., we call every assumption axiom but treat them as questionable postulates. 

The main result of this paper states that {\it the axioms above
  neither prove nor refute the existence of FTL particles.}  To prove
this statement, we construct two models (solutions) of the axioms
above such that one of them allows the existence of FTL particles but
the other does not.

The construction of one observer's worldview in which there can be FTL
particles is easy. The difficulty is to construct the worldviews of
all the observers moving slower than light relative to this observer
such that the principle of relativity holds for them.  To this, we
need to ensure that the same possible FTL particles can exist for all
the observers. This can be done, e.g., by allowing all the FTL speeds
including the infinite one to be a possible speed of a particle. For
details, see Section~\ref{sec-thm}.
\section{The language of our axiom system}
\label{sec-lang}

Let us now reconstruct the ideas of Section \ref{sec-inf} within a
framework of formal logic.  To formulate Einstein's original informal
postulates within first-order logic, first we have to fix a set of
basic symbols for the theory, i.e., what objects and relations between
them we will use as basic concepts. 

Here we are going to use a sightly modified framework of
\cite{synthese}.  We will use the following two-sorted language of
first-order logic parametrized by a natural number $d\ge 2$
representing the dimension of spacetime:
\begin{equation}
\{\, \B,\Q\,; \IOb, \Ph,+,\cdot,\le,\W\,\},
\end{equation}
where $\B$ (bodies) and $\Q$ (quantities) are the two sorts, $\IOb$
(inertial observers) is a one-place relation symbol and $\Ph$ (light
signal emitted by) is a two-place relation symbols of sort $\B$, $+$
and $\cdot$ are two-place function symbols of sort $\Q$, $\le$ is a
two-place relation symbol of sort $\Q$, and $\W$ (the worldview
relation) is a $d+2$-place relation symbol the first two arguments of
which are of sort $\B$ and the rest are of sort $\Q$.

Relations $\IOb(m)$ and $\Ph(p,b)$ are translated as ``\textit{$m$ is
  an inertial observer},'' and ``\textit{$p$ is a light signal emitted
  by body $b$},'' respectively. To speak about coordinatization, we
translate $\W(k,b,x_1,x_2,\ldots,x_d)$ as ``\textit{body $k$
  coordinatizes body $b$ at space-time location $\langle x_1,
  x_2,\ldots,x_d\rangle$},'' (i.e., at space location $\langle
x_2,\ldots,x_d\rangle$ and instant $x_1$). 

{\bf Quantity  terms} are the variables of sort $\Q$ and what can be
built from them by using the two-place operations $+$ and $\cdot$,
{\bf body terms} are only the variables of sort $\B$.
$\IOb(m)$, $\Ph(p,b)$, $\W(m,b,x_1,\ldots,x_d)$, $x=y$, and $x\le y$
where $m$, $p$, $b$, $x$, $y$, $x_1$, \ldots, $x_d$ are arbitrary
terms of the respective sorts are so-called {\bf atomic formulas} of
our first-order logic language. The {\bf formulas} are built up from
these atomic formulas by using the logical connectives \textit{not}
($\lnot$), \textit{and} ($\land$), \textit{or} ($\lor$),
\textit{implies} ($\rightarrow$), \textit{if-and-only-if}
($\leftrightarrow$) and the quantifiers \textit{exists} ($\exists$)
and \textit{for all} ($\forall$).

We use the notation $\Q^n$ for the set of all $n$-tuples of elements
of $\Q$. If $\vx\in \Q^n$, we assume that $\vx=\langle
x_1,\ldots,x_n\rangle$, i.e., $x_i$ denotes the
$i$-th component of the $n$-tuple $\vx$. Specially, we write $\W(m,b,\vx)$ in
place of $\W(m,b,x_1,\dots,x_d)$, and we write $\forall \vx$ in place
of $\forall x_1\dots\forall x_d$, etc.

We use first-order logic set theory as a meta theory to speak about model
theoretical terms, such as models, validity, etc.  The {\bf models} of
this language are of the form
\begin{equation}
{\mathfrak{M}} = \langle \B, \Q;
\IOb_\mathfrak{M},\Ph_\mathfrak{M},+_\mathfrak{M},\cdot_\mathfrak{M},\le_\mathfrak{M},\W_\mathfrak{M}\rangle,
\end{equation}
where $\B$ and $\Q$ are nonempty sets, $\IOb_\mathfrak{M}$ is a unary
relation on $\B$, $\Ph_\mathfrak{M}$ is a binary relation on $\B$,
$+_\mathfrak{M}$ and $\cdot_\mathfrak{M}$ are binary operations and
$\le_\mathfrak{M}$ is a binary relation on $\Q$, and $\W_\mathfrak{M}$
is a subset of $\B\times \B\times \Q^d$.  Formulas are interpreted
in $\mathfrak{M}$ in the usual way.  For precise definition of the
syntax and semantics of first-order logic, see, e.g., \cite[\S
  1.3]{CK}, \cite[\S 2.1, \S 2.2]{End}. 

We denote that formula/statement $\varphi$ is {\bf valid} in model
$\mathfrak{M}$ by $\mathfrak{M}\models\varphi$.  A set of formulas
$\Sigma$ {\bf logically implies} formula $\varphi$, in symbols
$\Sigma\models\varphi$, iff (if and only if)  $\varphi$ is
valid in every model of $\Sigma$.

\section{Axioms for special relativity}
\label{sec-ax}

In this section, we formulate Einsteins original axioms in our
first-order logic language above.  Einstein has assumed two postulates
in his famous 1905 paper \cite{einstein}. The first was the principle
of relativity, which goes back to Galileo, see, e.g., \cite{galileo}
or \cite[pp.176-178]{taylor-wheeler}, and it roughly states that
inertial observers are indistinguishable from each other by
physical experiments, see, e.g., Friedman \cite[\S 5]{friedman}.

Principle of relativity strongly depends on the language in which it
is formalized, see Remark~\ref{rem-SPR}. So to introduce a general
version of the principle of relativity (and not just a kinematic one),
let $\mathcal{L}$ be a many-sorted first-order logic language (of
spacetime theory) containing at least sorts $\B$ and $\Q$ of our
language of Section \ref{sec-lang} and a unary relation $\IOb$ on sort
$\B$.  In this paper, $\mathcal{L}$ will be the language of Section
\ref{sec-lang}, except in the introduction of \ax{SPR_\mathcal{F}} below
and in Proposition~\ref{prop-aut} way below.

Let $\mathcal{F}$ be a set of formulas of $\mathcal{L}$ with at most
one free variable of sort $\B$. This set of formulas $\mathcal{F}$
will play the role of ``laws of physics'' in the formulation of the
principle of relativity theory. The free variable of sort $\B$ is used
to evaluate these formulas on observers and to check whether they are
valid or not according to the observer in question.  We call
$\mathcal{F}$ {\bf set of (potential) laws}.  Now we can formulate a
principle of relativity for each set of laws $\mathcal{F}$ as the
following axiom schema:
\begin{description}
\item[\underline{\ax{SPR_{\mathcal{F}}}}]\label{spr} A potential law of nature
  $\varphi\in \mathcal{F}$ is either true for all the inertial
  observers or false for all of them:
\begin{equation*}
\big\{\, \IOb(m)\land\IOb(k)\rightarrow
\big[\varphi(m,\vx)\leftrightarrow \varphi(k,\vx)\big]
\::\:\varphi\in\mathcal{F}\,\big\}.\footnote{That is, if $m$ and $k$
  are inertial observers, potential law $\varphi$ holds for $m$ with
  parameters $\vx$ if and only if it holds for $k$ with the same
  parameters.}
\end{equation*}
\end{description}
Now for all set of laws $\mathcal{F}$, we have a principle of
relativity.  Let us highlight two important cases. 

We call {\bf strong principle of relativity}, and denote by
\ax{SPR^+}, the one when $\mathcal{F}$ is the set of all formulas of
$\mathcal{L}$ with at most one free variable of sort $\B$, see also
\cite[p.84]{MPhd}. \ax{SPR^+} is implied by existence of
automorphisms of the model between any two inertial observers, see
Proposition~\ref{prop-aut} and Theorem~2.8.20 in \cite{MPhd}.  If
$\mathcal{F}\subseteq\mathcal{G}$, then \ax{SPR_{\mathcal{G}}} is
stronger than \ax{SPR_{\mathcal{F}}}, i.e.,
$\ax{SPR_{\mathcal{G}}}\models\ax{SPR_{\mathcal{F}}}$. Specially,
\ax{SPR^+} is stronger than any other \ax{SPR_{\mathcal{F}}}.

We call {\bf existential principle of relativity}, and denote by
\ax{SPR_\exists}, the one when $\mathcal{F}$ is the set of existential
formulas of our language with at most one free variable of sort
$\B$. The importance of \ax{SPR_\exists} is that the existential
formulas are in some sense the experimentally verifiable statements of
a physical theory. Similarly, universal formulas correspond to the
experimentally refutable statements of the theory. Analogously to
\ax{SPR_\exists}, we could introduce a universal version
\ax{SPR_\forall}, too. However, it is straightforward to show that
\ax{SPR_\forall} and \ax{SPR_\exists} are logically equivalent by
interchanging every universal formula $\varphi(h,\vx) \equiv \forall
\vu\; \psi(h,\vx,\vu)$ of \ax{SPR_\forall} and existential
formula $\varphi^*(h,\vx)\equiv \exists \vu\;
\neg\psi(h,\vx,\vu)$ of \ax{SPR_\exists}.

\begin{rem}\label{rem-SPR}
The richer the language $\mathcal{L}$ we choose to formulate the
principle of relativity the stronger our axiom schema \ax{SPR^+}
is. Therefore, if we add extra basic concepts (e.g., masses of bodies)
to our language it becomes more difficult to prove something (e.g.,
the existence of FTL particles) is consistent with \ax{SPR^+}.
\end{rem}

Let us also note that \ax{SPR_{\mathcal{F}}} contains infinite
statements if $\mathcal{F}$ is infinite. Consequently, \ax{SPR^+},
\ax{SPR_\exists}, and \ax{SPR_\forall} are all infinite lists of
statements.

Einstein's second postulate states that ``Any ray of light moves in
the “stationary” system of co-ordinates with the determined velocity
$c$, whether the ray be emitted by a stationary or by a moving body,''
see \cite{einstein}. We can easily formulate this statement in our
first-order logic frame. To do so, let us introduce the following two
concepts.\footnote{Relations $\timed$ and $\sqspace$ are definable if
  the structure of quantities is rich enough, e.g., if it is an
  ordered field, which will be implied by \ax{AxOField}, see
  p.\pageref{axofield}.}  The {\bf time difference} of coordinate
points $\vx,\vy\in\Q^d$ is defined as:
\begin{equation}
\timed(\vx,\vy)\de x_1-y_1. 
\end{equation}
The {\bf squared spatial distance}\footnote{Because we assume only
  very general assumptions about the quantities, see \ax{AxOField}
  below, to speak about the spatial distance of arbitrary two coordinate
  points, we have to use squared distance since it is possible that
  the distance of two points is not amongst the quantities. For
  example, the distance of points $\langle 0,0\rangle$ and $\langle
  1,1\rangle$ is $\sqrt{2}$. So in the field of rational numbers,
  $\langle 0,0\rangle$ and $\langle 1,1\rangle$ do not have distance
  but they have squared distance.} of $\vx,\vy\in\Q^d$ is defined as:
\begin{equation}
  \sqspace(\vx,\vy)\de (x_2-y_2)^2+\ldots+(x_d-y_d)^2.
\end{equation}

\begin{description}
\item[\underline{\ax{AxLight}}]\label{axlight} There is an inertial observer,
  according to whom, any light signal moves with the same velocity
  $c$, independently of the fact that which body emitted the signal.
  Furthermore, it is possible to send out a light signal in any
  direction (existing according to the coordinate system) everywhere:
\begin{multline*}
\exists mc\Big[\IOb(m)\land c>0\land\forall\vx\vy\,  \Big(\exists pb \big[ \Ph(p,b)\land \W(m,p,\vx)\\ \land
\W(m,p,\vy)\big] \leftrightarrow \sqspace(\vx,\vy)=
c^2\cdot\timed(\vx,\vy)^2\Big)\Big].\footnotemark
\end{multline*}
\end{description}
\footnotetext{That is, there are $m$ and $c$, such that $m$ is an inertial observer, $c$ is a positive quantity, and for all coordinate points $\vx$ and $\vy$ there is a light signal $p$ emitted by body $b$ coordinatized at $\vx$ and $\vy$ by observer $m$ if and only if equation $\sqspace(\vx,\vy)=c^2\cdot\timed(\vx,\vy)^2$ holds.}

Einstein assumed without postulating it explicitly that the structure
of quantities is the field of real numbers.  We make this postulate
more general by assuming only the most important algebraic properties
of real numbers for the quantities. 

\begin{description}\label{axofield}
\item[\underline{\ax{AxOField}}]
 The quantity part $\langle \Q,+,\cdot,\le \rangle$ is an ordered field, i.e.,
\begin{itemize}
\item  $\langle\Q,+,\cdot\rangle$ is a field in the sense of abstract
algebra; and
\item 
the relation $\le$ is a linear ordering on $\Q$ such that  
\begin{itemize}
\item[i)] $x \le y\rightarrow x + z \le y + z$ and 
\item[ii)] $0 \le x \land 0 \le y\rightarrow 0 \le xy$
holds.
\end{itemize}
\end{itemize}
\end{description}

Axiom \ax{AxOField} not only makes our theory more general, but also
opens a new research area investigating which algebraic properties of
numbers are needed by different spacetime theories, see also
\cite{wnst}.  An importance of this research area (as well as using
\ax{AxOField} instead of the field of real numbers) lies in the fact
that we cannot experimentally decide whether the structure of physical
quantities is really isomorphic to the field of real numbers or not.

To axiomatize special relativity based on Einstein's original
postulates, we have to explicitly state one more axiom which was
assumed implicitly by Einstein.  This axiom connects the
worldviews of different inertial observers by saying that all
observers coordinatize  the same ``external" reality (the same set of
events).  By the {\bf event} occurring for observer $m$ at coordinate
point $\vx$, we mean the set of bodies $m$ coordinatizes at $\vx$:
\begin{equation}
\ev_m(\vx)\de\{ b : \W(m,b,\vx)\}.
\end{equation}

\begin{description}\label{axev}
\item[\underline{\ax{AxEv}}]
All inertial observers coordinatize the same set of events:
\begin{equation*}
\IOb(m)\land\IOb(k)\rightarrow \exists \vy\, \forall b
\big[\W(m,b,\vx)\leftrightarrow\W(k,b,\vy)\big].
\end{equation*}
\end{description}
From now on, we will use $\ev_m(\vx)=\ev_k(\vy)$ to abbreviate the
subformula $\forall b\,[\W(m,b,\vx)\leftrightarrow\W(k,b,\vy)]$ of
\ax{AxEv}.

Basically we are ready for formulating an axiom system capturing
Einstein's special theory of  relativity within our framework of
first-order logic. Nevertheless, let us introduce three more
simplifying axioms.

To avoid trivial models, we also assume that there are inertial
observers moving relative to each other.
\begin{description}\label{axthexp}
\item[\underline{\ax{AxThExp}}] Inertial observers can move with any
  speed less than the speed of light:
\begin{multline*}
\exists h\, \IOb(h)\land \forall m \vx\vy \,
\Big(\IOb(m)\land \sqspace(\vx,\vy)<\ls_m^2\cdot\timed(\vx,\vy)^2
\\ \rightarrow \exists k \big[\IOb(k)\land
\W(m,k,\vx)\land\W(m,k,\vy)\big]\Big).
\end{multline*}
\end{description}

\begin{description}\label{axself}
\item[\underline{\ax{AxSelf}}]
Any inertial observer is stationary relative to himself:
\begin{equation*}
\IOb(m)\rightarrow \forall \vx\big[\W(m,m,\vx) \leftrightarrow
x_2=\ldots=x_d=0\big].
\end{equation*}
\end{description}

Axiom \ax{AxSelf} makes it easier to speak about the motion of
inertial observers since it identifies the observers with their
time-axises. So instead of always referring to the time-axises of
inertial observers we can speak about their motion directly.

Our last axiom is a symmetry axiom saying that
observers use the same units of measurement.
\begin{description}\label{axsymd}
\item[\underline{\ax{AxSymD}}]
Any two inertial observers agree as to the spatial distance between
two events if these two events are simultaneous for both of them; and
the speed of light is 1 for all observers:
\begin{multline*}
\IOb(m)\land\IOb(k) \land x_1=y_1\land
x'_1=y'_1\land \ev_m(\vx)=\ev_k(\vx')\\ \land
\ev_m(\vy)=\ev_k(\vy')\rightarrow \sqspace(\vx,\vy)=\sqspace(\vx',\vy'),
\text{ and }\\
\IOb(m)\rightarrow\exists
pb \big[\Ph(p,b)\land\W(m,p,0,\ldots,0)\land\W(m,p,1,1,0,\ldots,0)\big].
\end{multline*}
\end{description}

Axiom \ax{AxSymD} makes life easier (it simplifies the formulation of
our theorems) because we do not have to consider situations such as
when one observer measures distances in meters while another observer
measures them in feet.

Let us now introduce an axiom system of special relativity as the
collection of the axioms above:
\begin{equation*}\label{SR}
\ax{SR} \de \ax{SPR^+}+\ax{AxLight}+\ax{AxOField}+
\ax{AxEv}+\ax{AxThExp}+\ax{AxSelf}+ \ax{AxSymD}.
\end{equation*}

In our axiomatic approach, we usually use the following version of the
light axiom, which follows from 
\ax{SPR^+} and \ax{AxLight},
see Proposition~\ref{prop-axph}:
\begin{description}\label{axph}
\item[\underline{\ax{AxPh}}] For any inertial observer, the speed of
  light is the same everywhere and in every direction (and it is
  finite). Furthermore, it is possible to send out a light signal in
  any direction (existing according to the coordinate system)
  everywhere:
\begin{multline*}
\IOb(m)\rightarrow \exists c_m \Big[c_m>0\land \forall \vx\vy \, 
\Big(\exists pb \big[ \Ph(p,b)\land \W(m,p,\vx)\\\land \W(m,p,\vy)\big]
\leftrightarrow \sqspace(\vx,\vy)= c_m^2\cdot\timed(\vx,\vy)^2\Big)\Big].
\end{multline*}
\end{description}

Let us note here that \ax{AxPh} does not require (by itself) that the
speed of light is the same for every inertial observer.  It requires
only that the speed of light according to a fixed inertial observer is
a positive quantity which does not depend on the direction or the
location.

By \ax{AxPh}, we can define the {\bf speed of light} according to inertial
observer $m$ as the following binary relation:
\begin{multline*}
\ls(m,v)\defiff v>0  \land \forall \vx\vy\,\Big(
\exists pb \big[\Ph(p,b)\land \W(m,p,\vx)\\\land \W(m,p,\vy)\big]
\rightarrow \sqspace(\vx,\vy)= v^2\cdot\timed(\vx,\vy)^2\Big). 
\end{multline*}

By \ax{AxPh}, there is one and only one speed $v$ for every inertial observer
$m$ such that $\ls(m,v)$ holds. From now on, we will denote this unique speed by $\ls_m$.

Let us now prove that Einstein's light axiom formulated as
\ax{AxLight} and the principle of relativity \ax{SPR_{\mathcal{F}}} implies \ax{AxPh} if the
set of laws $\mathcal{F}$ contain a certain existential formula.

\begin{prop}\label{prop-axph}
Let $\mathcal{F}$ be set of laws containing formula $\exists
pb\, [\Ph(p,b)\land \W(h,p,\vx)\land\W(h,p,\vy)]$. Then
\begin{equation*}
\ax{SPR_{\mathcal{F}}}+ \ax{AxLight} + \ax{AxOField}\models \ax{AxPh}\land  \forall mk
\big[\IOb(m)\land\IOb(k)\rightarrow \ls_m=\ls_k\big].
\end{equation*}
\end{prop}
\begin{proof}
By \ax{SPR_{\mathcal{F}}}, we get that 
\begin{multline*}
\IOb(m)\land\IOb(k)\rightarrow\Big(
\exists pb \big[ \Ph(p,b)\land \W(m,p,\vx)\land\W(m
,p,\vy)\big]\\ \leftrightarrow \exists pb \big[ \Ph(p,b)\land
\W(k,p,\vx)\land\W(k,p,\vy)\big]\Big).
\end{multline*}
By axiom \ax{AxLight}, there are an inertial observer $m$ and a
positive quantity $c$ such that
\begin{multline*}
\exists pb \big[ \Ph(p,b)\land \W(m,p,\vx)\land \W(m,p,\vy)\big]
\\
\leftrightarrow \sqspace(\vx,\vy)= c^2\cdot\timed(\vx,\vy)^2 
\end{multline*}
Therefore, for all inertial observer $m$ exists a light signal moving
through $\vx$ and $\vy$ iff $\sqspace(\vx,\vy)=
c^2\cdot\timed(\vx,\vy)^2$, i.e., formula
\begin{multline*}
\exists c \Big[ c>0 \land \forall m\vx\vy \,\Big( \IOb(m)\rightarrow 
\exists pb \big[ \Ph(p,b)\land \W(m,p,\vx)\\\land
\W(m,p,\vy)\big]  \leftrightarrow \sqspace(\vx,\vy)=
c^2\cdot\timed(\vx,\vy)^2\Big)\Big]
\end{multline*}
follows from \ax{SPR_{\mathcal{F}}} and \ax{AxLight}. This formula
implies both axiom \ax{AxPh} and formula $\forall mk\,[
\IOb(m)\land\IOb(k)\rightarrow \ls_m=\ls_k]$.
\end{proof}

Since the sets of laws corresponding to \ax{SPR^+} and
\ax{SPR_\exists} contains formula $\exists pb\,[ \Ph(p,b)\land
\W(m,p,\vx)\land\W(m,p,\vy)]$, Proposition~\ref{prop-axph} implies the following:
\begin{cor}\label{cor-axph}
\begin{align*}
\ax{SPR^+}+ \ax{AxLight} + \ax{AxOField}&\models \ax{AxPh} \land \big[
\IOb(m)\land\IOb(k)\rightarrow \ls_m=\ls_k\big],\\
\ax{SPR_\exists}\,+ \ax{AxLight} + \ax{AxOField}&\models \ax{AxPh} \land \big[
\IOb(m)\land\IOb(k)\rightarrow \ls_m=\ls_k\big].
\end{align*}
\end{cor}

We usually assume \ax{AxPh} instead of \ax{SPR^+} and \ax{AxLight}
because, by Corollary~\ref{cor-axph}, \ax{AxPh} is a more general
basic assumption, but it is strong enough to capture the kinematics of
special relativity, see, e.g., Theorem~\ref{thm-poi}. Therefore, let
us also recall the finite axiom system for special relativity used in
\cite{synthese}:
\begin{equation*}
\ax{SpecRel} \de \ax{AxPh}+\ax{AxOField}+
\ax{AxEv}+\ax{AxSelf}+ \ax{AxSymD}.\footnote{Because, models in which there are no inertial observers moving relative to each other are trivial but do not ruin the key theorems of special relativity, we do not include \ax{AxThExp} in \ax{SpecRel}.}
\end{equation*}

It is easy to see that there are great many models of \ax{SpecRel} which
are not models of \ax{SR}. Therefore, our theory \ax{SpceRel} is more
general than \ax{SR} by Corollary~\ref{cor-axph}, i.e., the following
is true:
\begin{cor}\label{cor-sr}
\begin{equation*}
\ax{SR}\models \ax{SpecRel}\enskip\text{ and }\enskip
\ax{SpecRel}\not\models\ax{SR}.
\end{equation*}
\end{cor}

In relativity theory, we are often interested in comparing the
worldviews of two different observers. To characterize the possible
relations between the worldviews of inertial observers, let us
introduce the {\bf worldview transformation} between observers $m$ and
$k$ (in symbols, $\w_{mk}$) as the binary relation on $\Q^d$
connecting the coordinate points where $m$ and $k$ coordinatize the
same events:
\begin{equation}\label{eq-ww}
\w_{mk}(\vx,\vy)\defiff \forall b\big[\W(m,b,\vx)\leftrightarrow \W(k,b,\vy)\big].
\end{equation}

Map $P:\Q^d\rightarrow\Q^d$ is called a {\bf Poincar\'e transformation} iff
it is an affine bijection having the following property
\begin{equation}
\timed(\vx,\vy)^2-\sqspace(\vx,\vy)=\timed(\vx',\vy')^2-\sqspace(\vx',\vy')
\end{equation}
for all $\vx,\vy,\vx',\vy'\in\Q^d$ for which  $P(\vx)=\vx'$ and $P(\vy)=\vy'$.

Theorem~\ref{thm-poi} shows that even our (more general) axiom system
\ax{SpecRel} perfectly captures the kinematics of special relativity
since it implies that the worldview transformations between inertial
observers are the same as in the standard non-axiomatic approaches.
\begin{thm}\label{thm-poi}
Let $d\ge3$. Assume \ax{SpecRel}. Then $\w_{mk}$ is a Poincar\'e
transformation if $m$ and $k$ are inertial observers.
\end{thm}

For the proof of Theorem~\ref{thm-poi}, see \cite{wnst}. For versions
of Theorem~\ref{thm-poi} using a similar but different axioms systems
of special relativity, see, e.g., \cite{pezsgo}, \cite{AMNSamples},
\cite{logst}.

\begin{cor}\label{cor-poi}
Let $d\ge3$. Assume \ax{SR}. Then $\w_{mk}$ is a Poincar\'e
transformation if $m$ and $k$ are inertial observers. \qed
\end{cor}
\noindent
The  {\bf worldline} of body $b$ according to observer $m$ is
defined as:
\begin{equation}
\wl_m(b)\de\{ \vx: \W(m,b,\vx)\}.
\end{equation}
\begin{cor}\label{cor-line}
Let $d\ge3$. Assume \ax{SpecRel} or \ax{SR}. The $\wl_m(k)$ is a
straight line if $m$ and $k$ are inertial observers.
\end{cor}

\section{Independence of FTL bodies of SR}
\label{sec-thm}
Before we show that the existence of FTL objects (bodies) is
independent of \ax{SR}, let us first prove a useful connection
between \ax{SPR^+} and the automorphism group of the models of our
language. 

\begin{prop}\label{prop-aut}
Let $\mathcal{L}$ be a language on which \ax{SPR^+} is formalized. Let
$\mathfrak{M}$ be a model of language $\mathcal{L}$. If for all
inertial observers $m$ and $k$, there is an automorphism of
$\mathfrak{M}$ fixing the quantities and taking $m$ to $k$, then
\ax{SPR^+} is valid in $\mathfrak{M}$.
\end{prop}
\begin{proof}
Let $\varphi(h,\vx)$ be an arbitrary formula with only free variables
$h$ of sort $\B$ and $\vx$ of sort $\Q$. To prove \ax{SPR^+}, we have
to prove that
\begin{equation}\label{eq-spr}
\mathfrak{M}\models\forall mk\vx\,\Big( \IOb(m) \land
\IOb(k)\rightarrow
\big[\varphi(m,\vx)\leftrightarrow\varphi(k,\vx)\big]\Big).
\end{equation}
Let $m$ and $k$ be inertial observers and let $Aut_{mk}$ be the
automorphism taking $m$ to $k$ and fixing the quantities. This
automorphism exists by our assumption. Let $\vx$ be such that
$\mathfrak{M}\models\varphi(m,\vx)$. Then
$\mathfrak{M}\models\varphi\big(Aut_{mk}(m),Aut_{mk}(\vx)\big)$ since
$Aut_{mk}$ is an automorphism of $\mathfrak{M}$. Since $Aut_{mk}(m)=k$
and $Aut_{mk}(\vx)=\vx$, we have that
$\mathfrak{M}\models\varphi(k,\vx)$. Therefore, the following half of
\eqref{eq-spr} holds
\begin{equation}
\mathfrak{M}\models \forall mk\vx\,\Big(
\IOb(m)\land\IOb(k)\rightarrow
\big[\varphi(m,\vx)\rightarrow\varphi(k,\vx)\big]\Big)
\end{equation}
since $m$ and $k$ were arbitrary inertial observers and $\vx$ was an
arbitrary sequence of quantities. By interchanging $m$ and $k$, we get
the converse direction. So $\mathfrak{M}\models\ax{SPR^+}$; and this is
what we wanted to prove.
\end{proof}
For model theoretic characterizations of \ax{SPR^+} based on existence
of automorphism connecting the worldviews of inertial observers, see
Theorem~2.8.20 in \cite{MPhd}.

Now let us prove that the existence of FTL objects (bodies) is
independent of special relativistic kinematics. To do so, let us
introduce a formula of our language stating that there is a
superluminal body.
\begin{description}\label{ftlbody}
\item[\underline{\ax{\exists FTLBody}}] There is an inertial observer
  according to who a body moving faster than the speed of light:
\begin{multline*}
\exists mb\vx\vy \,\big[\IOb(m)\land \W(m,b,\vx)\land
\W(m,b,\vy)\\\land \sqspace(\vx,\vy)>\ls_m^2\cdot \timed(\vx,\vy)^2\big].
\end{multline*}
\end{description}

\begin{thm}\label{thm-indep}
\begin{equation*}
\ax{SR}\not\models\ax{\neg\exists FTLBody}\enskip\text{ and } \enskip
\ax{SR}\not\models\ax{\exists FTLBody}
\end{equation*}
\end{thm}

In the proof of Theorem~\ref{thm-indep}, we will use the following
concepts. The {\bf identity map} is defined as:
\begin{equation}
\Id(\vx)\de\vx\enskip\text{ for all }\enskip\vx\in\Q^d.
\end{equation}
The {\bf time-axis} is defined as the following subset of the
coordinate system $\Q^d$:
\begin{equation}\label{eq-taxis} 
\taxis\de\{\,\vx\in\Q^d\::\: x_2=\ldots=x_d=0\,\}.
\end{equation}
We think of functions as special binary relations. Hence we compose
them as relations.  The {\bf composition} of binary relations $R$ and
$S$ is defined as:
\begin{equation}\label{rcomp}
{R \fatsemi S}\de \{\,\langle a,c\rangle\::\: \exists b\,
[R(a,b)\land S(b,c)] \,\}.
\end{equation}
Therefore, 
\begin{equation}\label{eq-fn}
(g\fatsemi f)(x)=f\big(g(x)\big)
\end{equation}
if $f$ and $g$ are functions. 
Let $H$ be a subset of $\Q^d$ and let $f:\Q^d\rightarrow\Q^d$ be a
map. The {\bf $f$-image} of set $H$ is defined as:
\begin{equation}
f[H]\de\{\,f(\vx)\::\:\vx\in H\,\}.
\end{equation}
The {\bf inverse} binary relation of $R$ is defined as:
\begin{equation}\label{rinv}
{R^{-1}}\de \{\,\langle a,b\rangle\::\:  R(b,a) \,\}.
\end{equation}

\begin{proof}
We are going to prove our statement by constructing two models
$\mathfrak{M}_1$ and $\mathfrak{M}_2$ of \ax{SR} such that in
$\mathfrak{M}_1$ there are FTL bodies and in $\mathfrak{M}_2$ there
are no FTL bodies.  There will be only a slight difference in the two
constructions. Therefore, we are going to construct the two models
simultaneously.

Let $\langle \Q, +,\cdot,\le\rangle$ be the ordered field of real
numbers. Let us introduce the following sets, which will be unary
relations on set $\B$ after we will define $\B$:
\begin{equation}\label{eq-IOb}
\IOb\de\{\,\text{Poincar\'e transformations of }\Q^d\,\},
\end{equation}
\begin{equation}\label{eq-B}
\LS \de \{\, \text{lines of slope 1 in }\Q^d\,\},
\enskip\text{ and }\enskip
\IB\de\{\,\text{lines in }\Q^d\,\}.
\end{equation}

Let binary relation $\Ph$ be defined as:
\begin{equation}
\Ph(p,b) \defiff p\in \LS \land b=p.
\end{equation}

\begin{figure}
\psfrag{mm}[bc][bc]{$m$}
\psfrag{m}[tl][tl]{$m$}
\psfrag{p}[tr][tr]{$p$}
\psfrag{id}[tl][tl]{$\Id$}
\psfrag{b}[br][br]{$b$}
\psfrag{x}[tl][tl]{$\vx$}
\psfrag{mx}[tl][tl]{$m(\vx)$}
\includegraphics[keepaspectratio,width=0.8\textwidth]{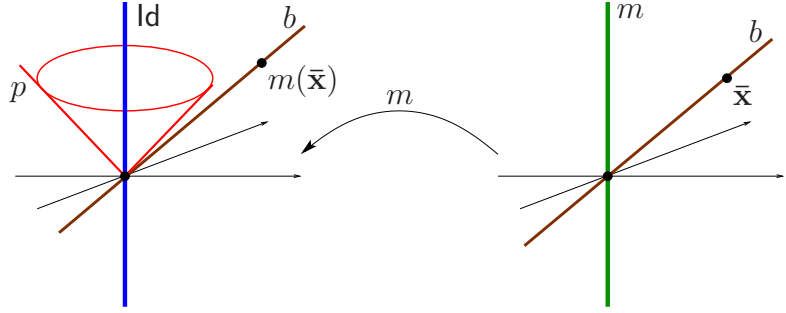}
\caption{Illustration for the proof of Theorem~\ref{thm-indep}}
\label{fig-w}
\end{figure}

The only difference between the construction of models
$\mathfrak{M}_1$ and $\mathfrak{M}_2$ is in the definition of the set
of bodies. In models $\mathfrak{M}_1$ and $\mathfrak{M}_2$, the
respective sets of bodies are defined as:
\begin{equation}
\B_1=\IOb\cup\LS\cup\IB\enskip \text{ and }\enskip \B_2=\IOb\cup\LS.
\end{equation}
First we are going to give the worldview of observer
$\Id$. Let 
\begin{equation}
\W(\Id,\Id,\vx)\defiff x_2=\ldots=x_d=0;
\end{equation} 
for any other inertial observer $m$, let
\begin{equation}\label{eq-Idm}
\W(\Id,m,\vx) \defiff \vx\in m[{\taxis}];
\end{equation}
and for any body $b$ which is not an  observer, i.e.,
$b\in\B\setminus\IOb$, let
\begin{equation}\label{eq-Idb}
\W(\Id,b,\vx)\defiff \vx\in b.
\end{equation}
Now the worldview of observer $\Id$ is given.  From the worldview of
$\Id$, we construct the worldview of another inertial observer $m$ as
follows:
\begin{equation}\label{ww-def}
\W(m,b,\vx)\defiff \W\big(\Id,b,m(\vx)\big)
\end{equation}
for all body $b\in\B$, see Figure~\ref{fig-w}.  Now models $\mathfrak{M}_1$ and
$\mathfrak{M}_2$ are given.

By the above definition of $\W$, if $m$ and $k$ are inertial
observers, then
\begin{equation}\label{eq-iob}
\W(m,k,\vx)\enskip \text{ holds iff }\enskip m(\vx)\in k[\taxis],
\end{equation}
and if $m\in\IOb$
and $b\in\B\setminus\IOb$, then 
\begin{equation}\label{eq-b}
\W(m,b,\vx)\enskip\text{ holds iff }\enskip
m(\vx)\in b.
\end{equation}
The worldview transformations between inertial
observers $m$ and $\Id$ is $m$, i.e., $\w_{m\Id}=m$ by equation
\eqref{ww-def}. Therefore, the worldview transformation between
inertial observers $m$ and $k$ is $m\fatsemi k^{-1}$, i.e., 
\begin{equation}\label{eq-w}
\w_{mk}=m\fatsemi k^{-1}
\end{equation}
since $\w_{mk}=\w_{m\Id}\fatsemi\w_{\Id k}$ and $\w_{\Id k}=(\w_{k\Id
})^{-1}$ by the definition of the worldview
transformation \eqref{eq-ww}. Specially, the worldview transformations between
inertial observers are Poincar\'e transformations in these models (as
Theorem~\ref{cor-poi} requires it). Hence
\begin{equation}\label{eq-bij}
\w_{mk} \text{ is a bijection for all inertial observers $m$ and $k$.}
\end{equation}

Now we are going to show that $\mathfrak{M}_1$ and $\mathfrak{M}_2$
 are models of axiom system \ax{SR}.  It is clear that axiom
 \ax{AxOField} is valid in these models since the ordered field of
 real numbers is an ordered field.

Axiom \ax{AxEv} is valid in models $\mathfrak{M}_1$ and
$\mathfrak{M}_2$ since worldview transformations between inertial
observers are bijections by \eqref{eq-bij}.

It is clear that axiom \ax{AxSelf} is valid in these models by the
definition of worldview relation $\W$ since
\begin{multline*}
\W(m,m,\vx)\stackrel{\eqref{eq-iob}}{\iff} m(\vx)\in
m[\taxis]\\\stackrel{\eqref{eq-IOb}}{\iff} \vx\in\taxis
\stackrel{\eqref{eq-taxis}}{\iff} x_2=\ldots=x_d=0.
\end{multline*}

It is clear that \ax{AxLight} is valid in models $\mathfrak{M}_1$ and
$\mathfrak{M}_2$ by the construction the worldview of inertial
observer $\Id$. Moreover, the speed of light is 1 for observer $\Id$.
Since Poincar\'e transformations take lines of slope 1 to lines of
slope 1, the speed of light is $1$ according to every inertial
observer, which is the second half of \ax{AxSymD}.

Any Poincar\'e transformation $P$ preserves the spatial distance of
points $\vx,\vy\in\Q^d$ for which $x_1=y_1$ and
$P(\vx)_1=P(\vy)_1$. Therefore, inertial observers agree as to the
spatial distance between two events if these two events are
simultaneous for both of them. We have already shown that the speed of
light is 1 according to each inertial observer in models  $\mathfrak{M}_1$ and
$\mathfrak{M}_2$. Consequently, axiom
\ax{AxSymD} is also valid in these models.

By Lemma~\ref{prop-aut}, to prove that axiom \ax{SPR^+} is valid in
models $\mathfrak{M}_1$ and $\mathfrak{M}_2$, it is enough to show
that, for all inertial observers $m$ and $k$, there is an automorphism
$Aut_{mk}$ fixing the quantities and taking $m$ to $k$.

Let us fix two arbitrary inertial observers $m$ and $k$. Let
$Aut_{mk}$ be the following map
\begin{equation}\label{eq-autiob}
Aut_{mk}(h)\de h\fatsemi
\w_{mk}=h\fatsemi m^{-1}\fatsemi k
\end{equation}
for all inertial observers $h$. 
\begin{equation}\label{eq-autb}
Aut_{mk}(q) \de q \enskip\text{ and }\enskip Aut_{mk}(b)\de\w_{mk}[b]
\end{equation}
for all quantity $q$ and for all body $b$ which is not an inertial
observer.  Since $\w_{mk}$ is a bijection, it is clear that $Aut_{mk}$
is also a bijection of $\mathfrak{M}_1$ and $\mathfrak{M}_2$. It is
clear that $Aut_{mk}(m)=k$ since $Aut_{mk}(m)=m\fatsemi m^{-1}\fatsemi
k=k$. It is also clear that $Aut_{mk}$ is fixing $\Q$ by its
definition.  Now we show that $Aut_{mk}$ is an automorphism. Since
$Aut_{mk}$ is fixing the quantities, $x\le y$, etc.\ hold iff
$Aut_{mk}(x)\le Aut_{mk}(y)$, etc. hold.

Let $h$ be a body.  We have to show that $\IOb(h)$ holds iff
$\IOb\big(Aut_{mk}(h)\big)$ holds. This is so because
$Aut_{mk}(h)=h\fatsemi m^{-1}\fatsemi k$ is a Poincar\'e
transformation iff $h$ is a Poincar\'e transformation.

Let $p$ and $b$ be bodies.  We have to show that $\Ph(p,b)$ holds iff
$\Ph\big(Aut_{mk}(p),Aut_{mk}(b)\big)$ holds. Relation $\Ph(p,b)$
holds iff $p\in\LS$ and $p=b$ hold. Since transformation $\w_{mk}$ is
a bijection taking lines of slope 1 to lines of slope 1 (since it is a
Poincar\'e transformation), $p\in\LS$ iff
  $Aut_{mk}(p)\in Aut_{mk}[\LS]$ because
  $Aut_{mk}[\LS]=\LS$ and $Aut_{mk}(p)=\w_{mk}(p)$. Therefore, relation $\Ph(p,b)$ holds iff
$\Ph\big(Aut_{mk}(p),Aut_{mk}(b)\big)$ holds since $p=b$ holds iff
$\w_{mk}[p]=\w_{mk}[b]$ holds.

Let $h$ be an inertial observer, $b$ be a body and $\vx$ be a
coordinate point. We have to show that relation $\W(h,b,\vx)$ holds iff
relation $\W\big(Aut_{mk}(h),Aut_{mk}(b),Aut_{mk}(\vx)\big)$ holds.  There are two
cases: $b\in\IOb$, and  $b\in \B\setminus\IOb$. 
If $b\in \IOb$, then
\begin{multline*}
\W\big(Aut_{mk}(h),Aut_{mk}(b),Aut_{mk}(\vx)\big)
\stackrel{\eqref{eq-autiob}}{\iff} \W(h\fatsemi \w_{mk}, b\fatsemi
\w_{mk},\vx) \\\stackrel{\eqref{eq-iob}}{\iff} (h\fatsemi
\w_{mk})(\vx)\in (b\fatsemi \w_{mk})[\taxis]\stackrel{\eqref{eq-fn}}{\iff}
\w_{mk}\big(h(\vx)\big)\in
\w_{mk}\big[b[\taxis]\big]\\\stackrel{\eqref{eq-bij}}{\iff} h(\vx)\in
b[\taxis] \stackrel{\eqref{eq-iob}}{\iff} \W(h,b,\vx).
\end{multline*}
If $b\in \B\setminus\IOb$, then
\begin{multline*}
\W\big(Aut_{mk}(h),Aut_{mk}(b),Aut_{mk}(\vx)\big)\\\stackrel{\eqref{eq-autb}}{\iff}
\W(h\fatsemi \w_{mk}, \w_{mk}[b],\vx) \stackrel{\eqref{eq-b}}{\iff}
(h\fatsemi\w_{mk})(\vx)\in \w_{mk}[b]\\\stackrel{\eqref{eq-fn}}{\iff}
\w_{mk}\big(h(\vx)\big)\in\w_{mk}[b]\stackrel{\eqref{eq-bij}}{\iff} h(\vx)\in
b\stackrel{\eqref{eq-b}}{\iff} \W(h,b,\vx).
\end{multline*}
Therefore, by Lemma~\ref{prop-aut}, axiom \ax{SPR^+} is valid in the
models.

The $\exists h\, \IOb(h)$ part of axiom \ax{AxThExp} is valid in models
$\mathfrak{M}_1$ and $\mathfrak{M}_2$, since there are Poincar\'e
transformations (e.g., $\Id$ is one).  Since \ax{SPR^+} is valid in
$\mathfrak{M}_1$ and $\mathfrak{M}_2$, it is enough to prove the
second part of \ax{AxThExp} only for observer $\Id$ instead of for all
inertial observer $m$.  To do so, let $\vx$ and $\vy$ be coordinate
points such that $\sqspace(\vx,\vy)<\timed(\vx,\vy)^2$. It is easy to
see that there is a Poincar\'e transformation $k$ for which
$\vx,\vy\in k[\taxis]$.  Therefore, by the definition of worldview
transformation \eqref{eq-ww} and by \eqref{eq-Idm}, $\W(\Id,k,\vx)$
and $\W(\Id,k,\vy)$ hold for inertial observer $k$ as \ax{AxThExp}
requires it.

It is clear by the constructions of the models that in $\mathfrak{M}_1$
there are FTL bodies and in $\mathfrak{M}_2$ there is no
FTL body. So neither \ax{\exists FTLBody} nor
$\ax{\neg\exists FTLBody}$ follows from \ax{SR}, and this is what we
wanted to prove.
\end{proof}

By Theorem~\ref{thm-indep}, the existence of FTL bodies is independent
of the theory of special relativity. Specially, it is consistent with
axiom system \ax{SR} that there are FTL objects (bodies). Because
\ax{SR} is based on the formulation of Einstein's original postulates
of special relativity, we have formally proved that the kinematics of
Einstein's special theory of relativity is consistent with the
existence of FTL particles.

That \ax{\exists FTLBody} is logically independent of \ax{SR} is
completely analogous to that the axiom of parallels is independent of
the rest of the axioms of Euclidean geometry or to that the continuum
hypothesis is independent of Zermelo-Fraenkel set theory. It means
that we can (consistently) extend theory \ax{SR} both with assumption
\ax{\exists FTLBody} or \ax{\neg\exists FTLBody}. However, until we
have a good reason to assume either of them, it is better to work
without these extra assumptions.

\section{Sending information back to the past}
\label{sec-ftl}

One of the nontrivial consequences of the possibility of sending out
FTL particles is that, in some sense, we can send signals back to the
past.  More precisely, when two inertial observers are moving relative
to each other, then both observers can send out a superluminal signal
to the other such that, according to the other, the event of sending
out of the signal happened later than the event of receiving of the
signal. To illustrate this fact, let us consider two observers
moving relative to each other with speed $0.5\, c$. Let us call one of
these observers ``stationary observer'' and the other ``moving
observer.''  The worldview of the ``moving observer'' is drawn into
the worldview of the ``stationary observer'' in
Figure~\ref{fig-signal}. Let the ``moving observer'' send out a signal
with speed $3\, c$ at event A. It can be read from the figure that
this signal reach the ``stationary observer'' at event B which is
earlier than event A according to the ``stationary observer.''

\begin{figure}
\psfrag{1}{1}
\psfrag{2}{2}
\psfrag{3}{4}
\psfrag{5}{5}
\psfrag{6}{6}
\psfrag{A}[tl][tl]{A}
\psfrag{B}[br][br]{B}
\psfrag{so}{``stationary observer''}
\psfrag{mo}{``moving observer''}
\psfrag{S}[tr][tr]{Signal}
\includegraphics[keepaspectratio,width=0.8\textwidth]{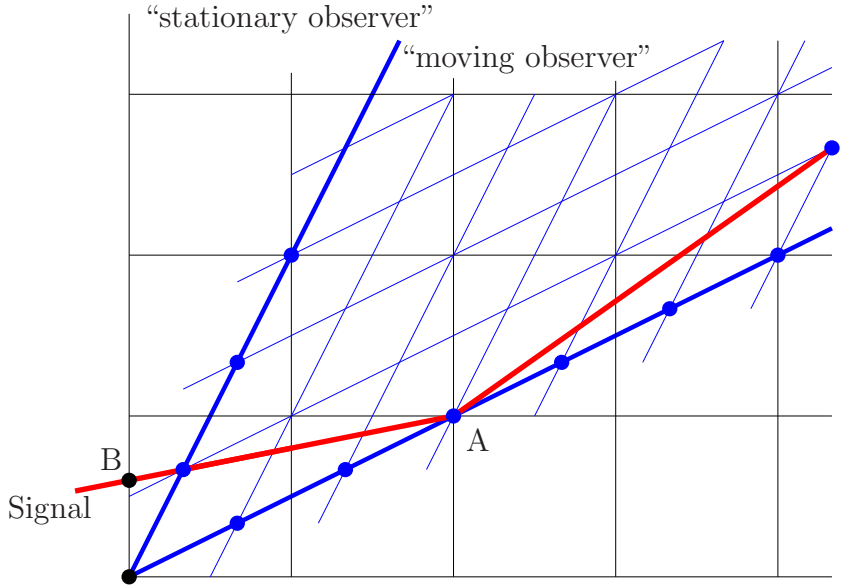}
\caption{Sending a superluminal signal back to the past: a signal of
  speed $3\, c$ which is sent out by the ``moving observer'' (moving
  with speed $0.5\, c$) at event A reaches the ``stationary observer''
  at event B which is earlier according to the ``stationary
  observer.''}
\label{fig-signal}
\end{figure}

If we can send bodies (particles) back to the past, then (in
  principle) we can also send any information back bit by bit. This
can easily be done by sending or not sending a signal back in time at
regular intervals; and the lack of signal means 0 and the presence of
the signal means 1 in the given bit. Using this protocol, we can send
arbitrary long information back to the past.

The possibility of sending back signals or even particles to the past
foreshadows the possibility of deriving a logical contradiction. In the
next section, we are going to show why this is not necessarily so even
if we leave the language of the kinematics.

\section{Why does not the possibility of sending
 particles back to the past lead to a logical contradiction?}
\label{sec-no}

Let us now illustrate with one possible informal explanation why the
existence of FTL objects does not necessarily lead to a logical
contradiction, let us consider the following well-known thought
experiment.  Let a tachyon cannon fire a superluminal projectile which
is reflected back from a moving mirror such that it hits the cannon
right before it has shot the projectile out, see
Figure~\ref{fig-cannon}.  The argument usually continues as
follows. If the cannon shoots the projectile at event A, the
projectile destroys the cannon at event C before event A. Therefore,
the cannon cannot shoot the projectile at event A.  This argument
seemingly contradicts the plausible assumption that our cannon can
shoot its FTL projectile anytime at any direction. So at first sight
we have derived a logical contradiction from the existence of FTL
particles.

However, if we think it over more carefully, we see that when we
calculated (speculated) the trajectory of the tachyon beam we did not
take into consideration of the causal effect of the projectile coming
back from the ``future.'' If we take this casual effect into
consideration, we can tell a logically consistent story in which the
tachyon cannon hits itself before the ignition of the projectile. For
example, we send out the mirror; and before we shoot the cannon to
destroy itself, a tachyon projectile hits the cannon (at event C) and
because of this damage the cannon shoots a tachyon projectile (at
event A) which reflects back from the mirror (at event B) so that it
becomes the one that hit the cannon (at event C).\footnote{According
  to the ``mirror'' event A is also later than event B. So it would be
  better to tell the story with another cannon in place of the mirror
  which shoots a projectile with speed 5 $c$ in the direction of the
  first cannon when it detects the projectile of the first cannon at
  event B.}  This story is a perfectly consistent story of the self
shooting cannon and a possible candidate of the resolution of this
``cannon paradox.''

This informal explanation is based on Novikov's self-consistency
principle, see, e.g., Lossev--Novikov \cite{Jinn}, which is a
standard way for resolving the causal paradoxes risen by closed
timelike curves in general relativity. For other ideas and
  informal explanations resolving the causal paradoxes of FTL
  particles, see, e.g., Arntzenius~\cite{Arntzenius},
Recami~\cite{recami-ftl}, \cite{Recami09}, Selleri~\cite{selleri-ftl}.

\begin{figure}
\psfrag{L}[bl][bl]{Tachyon Cannon}
\psfrag{M}[tl][tl]{Mirror}
\psfrag{A}[tr][tr]{A}
\psfrag{B}[tl][tl]{B}
\psfrag{C}[tr][tr]{C}
\includegraphics[keepaspectratio,width=0.8\textwidth]{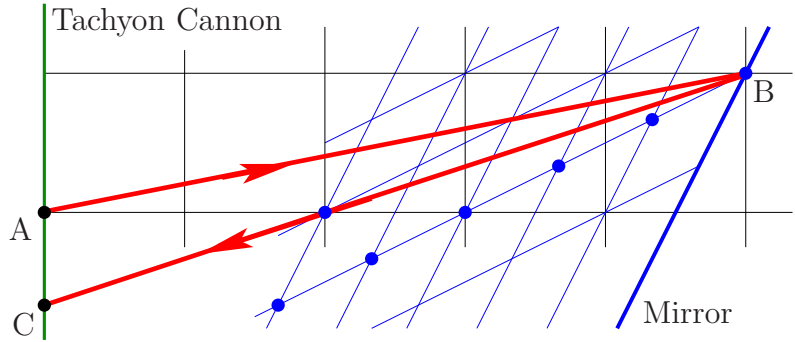}
\caption{Self shooting tachyon cannon: if a tachyon cannon fires an
  FTL tachyon beam (at event A) of speed 5 $c$ to a mirror moving with
  0.5 $c$ relative to the cannon, then the reflection of the tachyon
  beam will hit the cannon (at event C) before it have shot the
  beam.}
\label{fig-cannon}
\end{figure}

Let us note that the ``cannon paradox'' story above was told only in a
framework of kinematics, in which we cannot calculate (just speculate)
the effects of the collisions.  Within the kinematics of special
relativity the existence of FTL objects are perfectly consistent with
the theory, see Theorem~\ref{thm-indep} at
p.\pageref{thm-indep}. However, kinematics does not tell us anything
about the interactions of particles. Therefore, within kinematics we
can only say that the projectile meets the cannon at event C, but we
cannot say anything about the ``damage'' the projectile makes to the
cannon.

It is important to use axiomatic theories for resolving these
paradoxes, because this way we can give mathematical proofs for the
consistency and not just convincing but informal arguments.

\section{Anyway, inertial observers cannot move FTL}
We have seen that the kinematics of special relativity says nothing
about the existence of FTL particles in general. It is important to
note that even \ax{SpecRel} (which is more general than \ax{SR} by
Proposition~\ref{prop-axph}) implies that no inertial observer can
move faster than or with the speed of light if $d\ge 3$, see
\cite[Thm.2.1]{synthese}. For similar results on the impossibility of
existence of superluminal observers, see, e.g., \cite{pezsgo},
\cite{AMNSamples}, \cite{logst}, \cite{MNT}.

At first it may sound plausible to assume that, if inertial
particles can move with a certain speed, then inertial observers can
also move with this speed.  This assumption together with \ax{SpecRel}
(or \ax{SR}) would imply that there are no superluminal particles
since \ax{SpecRel} (and hence \ax{SR}) implies that there are no
FTL inertial observers. The problem
with this ``natural'' assumption is that (together with \ax{SpecRel}
or \ax{SR}) it also implies that there are no particles moving with
the speed of light which directly contradict \ax{AxLight}. So
surprisingly this assumption is not plausible at all since it
contradicts special relativity.

\section{Concluding remarks}

We have axiomatically proved that the existence of superluminal
objects is independent of (hence consistent with) the kinematics of
special relativity, see Theorem~\ref{thm-indep}. 

A future task is to investigate, within a similar axiomatic
  framework, how far this independence result can be extended beyond
kinematics. For a similar independence result of relativistic
  particle dynamics using the axiomatic framework of
  \cite{dyn-studia}, \cite{msz-wku}, \cite[\S 5]{SzPhd}, see
  \cite{MSzFTL}.

It is important to continue this investigation within an axiomatic
framework of logic stating all the assumptions explicitly.
Otherwise, there is a danger that some of our tacit assumptions will
remain hidden and we will not see clearly the logical connection
between our theory and the possibility of the existence of FTL
particles.

\section{Acknowledgment}
I am grateful to Hajnal Andr\'eka, Judit X.\ Madar\'asz, and Istv\'an
N\'emeti for the interesting discussion on the subject and their
valuable comments and suggestions.  This research is
supported by the Hungarian Scientific Research Fund for basic research
grants No.~T81188 and No.~PD84093.

\bibliography{LogRelBib}
\bibliographystyle{plain}

\end{document}